\documentclass[11pt,onecolumn]{IEEEtran}

\usepackage[utf8]{inputenc} 
\usepackage[T1]{fontenc}
\usepackage{url}              

\usepackage[cmex10]{amsmath}  
\usepackage{mleftright}       
\mleftright                   

\usepackage{graphicx}         
\usepackage{booktabs}         

\IEEEoverridecommandlockouts
\usepackage{amsmath,amssymb,amsfonts,amsthm}
\usepackage{enumitem}
\usepackage{mathtools}
\usepackage{algorithmic}
\usepackage{textcomp}
\usepackage[table,xcdraw]{xcolor}
\usepackage{tikz}

\usepackage{subcaption}
\usepackage{multicol,multirow}
\usepackage[style=ieee,citestyle=numeric-comp,sorting=nty]{biblatex}
\newtheorem{theorem}{Theorem}
\newtheorem{remark}[theorem]{Remark}

\newtheorem{proposition}[theorem]{Proposition}
\newtheorem{example}{Example}
\newtheorem{lemma}[theorem]{Lemma}

\newtheorem{construction}{Construction}

\newtheorem{observation}[theorem]{Observation}

\usepackage{scalerel}

\newcommand{\bF}{\mathbb{F}}

\newcommand{\cF}{\mathcal{F}}

\newcommand{\cU}{\mathcal{U}}

\newcommand{\boldc}{\mathbf{c}}

\newcommand{\boldu}{\mathbf{u}}

\DeclareSymbolFont{bbold}{U}{bbold}{m}{n}
\DeclareSymbolFontAlphabet{\mathbbold}{bbold}

\begin{document}
	

 \title{ 
 $\varepsilon$-MSR Codes \\ for Any Set of Helper Nodes}

		\author{Vinayak Ramkumar,  Netanel Raviv, and Itzhak Tamo\\
		\thanks{
	 Vinayak Ramkumar is with the  Department of Electrical Engineering--Systems, Tel Aviv University, Tel Aviv, Israel. Netanel Raviv is with the Department of Computer Science and Engineering, Washington University in St. Louis, St. Louis, MO, USA. Itzhak Tamo is with the Department of Electrical Engineering--Systems, Tel Aviv University, Tel Aviv, Israel. e-mail:	vinram93@gmail.com, netanel.raviv@wustl.edu, tamo@tauex.tau.ac.il. 
	
	The work of Vinayak Ramkumar and Itzhak Tamo was supported by the European Research Council (ERC grant number 852953).}
	}
	\maketitle
	\begin{abstract}
Minimum storage regenerating (MSR) codes are a class of maximum distance separable (MDS) array codes capable of repairing any single failed node by downloading the minimum amount of information from each of the helper nodes. However, MSR codes require large sub-packetization levels, which hinders their usefulness in practical settings. This led to the development of another class of MDS array codes called $\varepsilon$-MSR codes, for which the repair information downloaded from each helper node is at most a factor of $(1+\varepsilon)$ from the minimum amount for some $\varepsilon > 0$. The advantage of $\varepsilon$-MSR codes over MSR codes is their small sub-packetization levels. In previous constructions of epsilon-MSR codes, however, several specific nodes are required to participate in the repair of a failed node, which limits the performance of the code in cases where these nodes are not available. In this work, we present a construction of $\varepsilon$-MSR codes without this restriction. For a code with~$n$ nodes, out of which~$k$ store uncoded information, and for any number~$d$ of helper nodes ($k\le d<n$), the repair of a failed node can be done by contacting any set of~$d$ surviving nodes. Our construction utilizes group algebra techniques, and requires linear field size. We also generalize the construction to MDS array codes capable of repairing $h$ failed nodes using $d$ helper nodes with a slightly sub-optimal download from each helper node, for all $h \le r$ and $k \le d \le n-h$ simultaneously.  
 \end{abstract}
	\pagestyle{plain}
	
	\section{Introduction}
	In large-scale distributed storage systems, data needs to be stored in a redundant manner to ensure that the failure of storage nodes does not result in data loss. The na\"{i}ve solution is employing data replication; however, erasure coding is preferable as it offers the same reliability with lower storage overhead. In particular, maximum distance separable (MDS) codes are preferable for providing maximum fault tolerance for a given storage overhead.

The failure of a single node is a common occurrence in distributed storage systems, and the focus of the literature on codes for distributed storage \cite{CIT-115} has primarily been on the efficient repair of a single failed node. Regenerating codes, introduced by Dimakis et al. \cite{DimGodWuWaiRam}, minimize the \textit{repair bandwidth}, which is the amount of information downloaded from helper nodes to repair a failed node. These codes are array codes, with each code symbol being a vector of length~$\ell$ over some finite field, where $\ell$ is the \textit{sub-packetization level} of the code. Each node stores a code symbol (vector). Node failure is equivalent to the erasure of a code symbol, and node repair corresponds to recovering this erased code symbol. The nodes that assist in the repair process are called \textit{helper nodes}.  
	
	Among the class of regenerating codes, minimum storage regenerating (MSR) codes are of particular interest as they require minimum storage overhead for a given fault tolerance, i.e., they have the MDS property. The large sub-packetization requirement of MSR codes is one of the bottlenecks preventing the widespread usage of these codes in practical distributed storage systems. To tackle this issue, a class of MDS array code called $\varepsilon$-MSR codes was proposed [23].
 These codes enable the repair of a failed node by downloading from each helper node at most $(1+\varepsilon)$ times the optimal amount for some positive $\varepsilon$. It is shown in \cite{RawTamGur_epsilonMSR1} that small sub-packetization levels suffice for $\varepsilon$-MSR codes. For a constant number of parity nodes, the sub-packetization level of these codes scales only logarithmically with code length, which amounts to a doubly exponential savings in comparison to the best-known MSR codes \cite{LiWangHuYe}. However, previously known constructions of $\varepsilon$-MSR codes
 have the limitation that  
 either all the remaining nodes are contacted for the repair of a failed node  \cite{RawTamGur_epsilonMSR1}, or contacting a subset of remaining nodes suffices only if they include some specific nodes \cite{GuruswamiLJ20}.
 In short, these constructions 
 mandate certain nodes to
 serve as helper nodes during the repair process.
 The main goal of this paper is to construct $\varepsilon$-MSR codes that overcome this restriction. We obtain such an \(\varepsilon\)-MSR code by constructing a new MDS array code, which is then combined with a second code having a large normalized minimum distance. This new MDS has certain desirable repair properties, and it is constructed using group algebra techniques. 
 

 The remainder of the paper is organized as follows. In Section~\ref{sec:background} we provide the necessary background for $\varepsilon$-MSR codes and motivate our problem setting. In Section~\ref{sec:eps_msr} we present our construction of $\varepsilon$-MSR codes for which any arbitrary set of $d$ surviving nodes can serve as helper nodes for the repair of a failed node. In Section~\ref{sec:simul}, we generalize our construction to obtain MDS array codes that can repair any $h$ failed nodes using any $d$ helper nodes with a near-optimal repair bandwidth, for all $ h \le r$ and $k\le d \leq n-h$ simultaneously.

	\section{Notations and Background} \label{sec:background}
We begin this section by introducing notations and briefly describing MDS array codes and \(\varepsilon\)-MSR codes. For a positive integer \(n\), we use \([n]\) to denote the set of integers \(\{1, 2, \ldots, n\}\). All vectors are row vectors throughout this paper. For two integers \(a\) and \(b\), we use $a \bmod b$ to denote the remainder when \(a\) is divided by \(b\).

\subsection{MDS Array Codes}

An \((n,k,\ell)_\mathbb{F}\) array code is a code of length \(n\) and size \(|\mathbb{F}|^{\ell k}\) over the alphabet \(\mathbb{F}^\ell\), where \(\ell\) is the sub-packetization level of the code. 
Naturally, we view code symbols as vectors of length~$\ell$ over \(\mathbb{F}\). If the code is MDS, i.e., if each codeword possesses the property that the entire codeword can be recovered from any \(k\) out of these \(n\) code symbols, then it is termed an \((n,k,\ell)_\mathbb{F}\) MDS array code. These codes were developed to be used in distributed storage systems as follows. A file containing \(k\ell\) symbols over \(\mathbb{F}\) is encoded into a codeword comprised of \(n\) vectors of length \(\ell\) over \(\mathbb{F}\). These \(n\) vectors are stored on \(n\) different storage nodes, and thus, each node stores \(\ell\) symbols of \(\mathbb{F}\). The erasure of \(\ell\) symbols stored in a node is called a \textit{node failure}, and the process of recovering its contents by downloading symbols from the surviving nodes is called a \textit{node repair}.

When the array code is linear, it can be defined by a parity check matrix as follows.   
Let \(\mathcal{C} \subseteq \mathbb{F}^{n\ell}\) be a linear array code of block length \(n\) and sub-packetization level \(\ell\). It is defined by an \(r\ell \times n\ell\) parity-check matrix \(H\) over the finite field \(\mathbb{F}\) as \(\mathcal{C} = \{\mathbf{c} \in \mathbb{F}^{n\ell} \mid H\mathbf{c}^\intercal = \mathbf{0}\}\). Any codeword \(\mathbf{c} \in \mathcal{C}\) can be written as \(\mathbf{c} = (\mathbf{c}_1, \dots, \mathbf{c}_n) \in \mathbb{F}^{n\ell} \), with node \(i \in [n]\) storing \(\mathbf{c}_i \in \mathbb{F}^{\ell}\). 
Notice that every linear code has a so-called \textit{systematic form}, in which~$k$ nodes store uncoded information symbols and~$r=n-k$ nodes store \textit{parity symbols}.

Let \(H = [H_1 \dots H_n]\), with each sub-matrix \(H_i\) being an \(r\ell \times \ell\) matrix. For a subset \(S = \{i_1, \dots, i_{|S|}\} \subseteq [n]\), let 
\(H_S = [H_{i_1} \dots H_{i_{|S|}}] \in \mathbb{F}^{r\ell \times |S|\ell}\).
If $H_S$ is non singular for all~$S\subseteq[n]$ of size~$r$,
then \(H\) is a parity-check matrix of an \((n,k = n-r,\ell)_{\mathbb{F}}\) MDS array code. 
Throughout this paper, we define MDS array codes using parity-check matrices.
	
	\subsection{Minimum Storage Regenerating Codes}
An \((n,k,\ell)_\mathbb{F}\) MDS array code is said to be an \((n,k,d,\ell)_\mathbb{F}\) MSR code if any failed node can be repaired (i.e., its contents recreated) by contacting any arbitrary collection of \(d\) surviving nodes and downloading \(\frac{\ell}{d-k+1}\) symbols (over \(\mathbb{F}\)) from each of these \(d\) helper nodes. The number of helper nodes \(d\) will be referred to as the \textit{repair degree} and it satisfies \(k \le d < n\). Note that if \(d = k\), the repair is equivalent to downloading all the information from any subset of \(k\) nodes. This property holds automatically since the code is an MDS code. Hence, the nontrivial case is when \(k < d\).

In \cite{DimGodWuWaiRam}, it is shown that for MDS array codes, if \(d\) helper nodes are employed for node repair, then the repair bandwidth is at least \(\frac{d\ell}{d-k+1}\) symbols (over \(\mathbb{F}\)). In other words, the amount of information downloaded for node repair of MSR codes is the minimum possible for MDS array codes. MSR codes also have the load balancing property, i.e., the same amount of information is downloaded from each helper node.

For general MSR codes, during the repair process, helper nodes may have to transmit a function of the data stored in them. This can be problematic at times since it is possible that the amount of information transmitted is optimal; however, in order to compute it, one has to access and read the entire information stored on the node, which can be time-consuming. Therefore, a more desirable property is that the information that needs to be transmitted is exactly the information that is read from the node, and there is no need for any computation. An MDS array code is said to have the \textit{help-by-transfer} property
if the helper nodes access only the symbols that they transmit\footnote{This property is the same as the repair-by-transfer property in \cite{RawTamGur_epsilonMSR1} and the optimal-access property in \cite{YeBar_2}. In this work, we follow the help-by-transfer terminology in \cite{SasAgaKum}.}. 
Therefore, for MSR codes with the help-by-transfer property, only \(\frac{\ell}{d-k+1}\) symbols are accessed at each helper node during the repair.

	
	An early construction of MSR codes is of the product-matrix type \cite{RasShaKum_pm11}. Although the sub-packetization level requirement \(\ell = d - k + 1\) is small, this construction is limited to the low-rate regime of \(\frac{k}{n} \le \frac{1}{2} + \frac{1}{2n}\). In \cite{CadJafMalRamSuh}, the existence of MSR codes for all valid \((n, k, d)\) parameters as \(\ell \rightarrow \infty\) is shown. Several constructions of high-rate MSR codes are known in the literature, including those in \cite{TamWanBru, PapDimCad, WangTB16, SasAgaKum, RawatKV16, YeBar_1, GopFazVar, RavSilEtz, YeBar_2, SasVajKum_arxiv, LiTT18, VajhaBK23, LiWangHuYe}. 

Let \(s = d - k + 1\) and \(r = n - k\). For general \(k + 1 \le d \le n - 1\) parameters, the existence of MSR codes possessing the help-by-transfer property with sub-packetization level \(\ell = s^{\lceil \frac{n}{s} \rceil}\) is known \cite{RawatKV16, VajhaBK23, LiWangHuYe}. Even without the help-by-transfer requirement, the best-known construction \cite{LiWangHuYe} has a sub-packetization level of \(\ell = s^{\lceil \frac{n}{s+1} \rceil}\). We note that all the known MSR codes are linear, meaning that encoding and repair involve only linear operations.

Lower bounds on the sub-packetization level of linear MSR codes are derived in \cite{TamWanBru_access_tit, GoparajuTC14, AlrabiahG21, BabuVK22}. Specifically, it follows from the lower bound in \cite{BabuVK22} that \(\ell = s^{\lceil \frac{n}{s} \rceil}\) is the smallest possible sub-packetization level for linear MSR codes with the help-by-transfer property. For a constant number of parity nodes \(r\), the lower bound in \cite{AlrabiahG21} establishes that a sub-packetization level that is exponentially large in \(k\) is necessary for linear MSR codes. 

We refer readers to the survey in \cite{CIT-115} for a detailed discussion about various MSR code constructions and sub-packetization level lower bounds.

	\subsection{Small Sub-packetization Level}
Codes with small sub-packetization levels are preferable in practical systems for various reasons. A detailed description of the issues associated with large sub-packetization levels is provided in \cite{RawTamGur_epsilonMSR1}. It is argued that large sub-packetization levels reduce flexibility in selecting various system parameters. The difficulty in managing meta-data is another disadvantage of large sub-packetization levels, as listed in \cite{RawTamGur_epsilonMSR1}. Additionally, small sub-packetization levels lead to efficient degraded reads, as illustrated in \cite{RawTamGur_epsilonMSR1}. Another advantage of small sub-packetization levels is that they allow for the encoding of small files \cite{VajhaBK23}. In \cite{VajRamPur_Clay}, it is shown that large sub-packetization levels result in fragmented reads at helper nodes, leading to inferior disk read performance during repair. In summary, constructing codes with small sub-packetization levels is an important problem.

There have been multiple attempts in the literature to develop MDS array codes with small sub-packetization levels and near-optimal repair bandwidth, and the \(\varepsilon\)-MSR codes framework is one such effort. Early works investigating this problem include \cite{RasShaRam_Piggyback} and \cite{TamWanBru}. In \cite{RawTamGur_epsilonMSR1}, Rawat et al. presented an MDS array code with sub-packetization level \(\ell = (n - k)^{\tau}\) capable of repairing any failed node by accessing all remaining \(d = n - 1\) nodes, for any integer $1 \le \tau \le \lceil \frac{n}{n-k} \rceil -1$. The resulting repair bandwidth is \(\le (1 + \frac{1}{\tau})(\frac{n - 1}{n - k})\ell\). However, these codes require a very large field size.
In \cite{LiLiuTang}, Li et al. introduced multiple constructions of small sub-packetization level MDS array codes with similar repair bandwidths, again only for \(d = n - 1\). Some of these codes can be constructed over a field of size \(O(n)\). The existence of MDS array codes with \(\ell = a^{m + n - 1}\), where \(n - k = a^m\), and asymptotically optimal repair bandwidth for \(d = n - 1\) is shown in \cite{ChowVard}. In \cite{lin20231}, two MDS array codes with \(d < n - 1\), having small sub-packetization levels and near-optimal repair bandwidth, are constructed.

All the above-mentioned MDS array codes suffer from the defect that they do not possess the load balancing property. During repair, some helper nodes may have to contribute significantly more information than others, which may not be desirable in many system settings. Motivated by this, the authors of \cite{RawTamGur_epsilonMSR1} introduced the class of MDS array codes called \(\varepsilon\)-MSR codes, in which approximately the same amount of information is downloaded from each helper node during node repair.

	\subsection{$\varepsilon$-MSR Codes}  
	An \((n,k,\ell)_\mathbb{F}\) MDS array code is said to be an \((n,k,d,\ell)_\mathbb{F}\) \(\varepsilon\)-MSR code for \(\varepsilon > 0\) \cite{RawTamGur_epsilonMSR1} if the following property is satisfied:  
\begin{itemize}
    \item \(\varepsilon\)-optimal repair property: The content of any failed node can be recovered by contacting any arbitrary collection of \(d\) remaining nodes and downloading at most \((1+\varepsilon)\frac{\ell}{d-k+1}\) symbols over \(\mathbb{F}\) from each of these helper nodes. 
\end{itemize}
The amount of helper information downloaded from each of the \(d\) nodes is thus at most \((1+\varepsilon)\) times that of MSR codes. Setting \(\varepsilon = 0\) in the above description reduces to MSR codes.

For repair degree \(d = n-1\), an \(\varepsilon\)-MSR code is constructed in \cite{RawTamGur_epsilonMSR1} using an MSR code construction in \cite{YeBar_1} as a building block. The sub-packetization level of this code scales logarithmically with \(n\), provided the number of parity nodes \(r = n - k\) is a constant. This represents a substantial saving compared to MSR codes, which require \(\ell = \Omega(\exp(k))\). Additionally, this code requires a field size linear in \(n\). In \cite{GuruswamiLJ20}, an attempt was made to construct \(\varepsilon\)-MSR codes for the \(d < n-1\) case. However, the construction in \cite{GuruswamiLJ20} lacks the flexibility of choosing any \(d\) nodes among the surviving nodes as helper nodes and 
mandates a few nodes to be compulsorily contacted while repairing a failed node. Thus, strictly speaking, the code in \cite{GuruswamiLJ20} is not an \((n,k,d,\ell)_\mathbb{F}\) \(\varepsilon\)-MSR code. The advantage of having \(d < n-1\) is that the repair process can take place even if up to \(n-1-d\) nodes are temporarily unavailable. 
If some nodes must participate in the repair process (as in \cite{GuruswamiLJ20}), then the entire repair process is stalled even if one of them is unavailable.
In the current paper, we present the first construction of \(\varepsilon\)-MSR code without any compulsory helper node. See Table~\ref{table:compare} for a summary of the above discussion. 
\begin{table}[ht!]
\begin{center} 
\begin{tabular}{|c|c|c|} 
 \hline
 Reference& Range of $d$ & Number of compulsory helper nodes  \\
 \hline 
 Rawat et al. \cite{RawTamGur_epsilonMSR1} & $d=n-1$ & $n-1$ \\ 
 \hline 
 Guruswami et al. \cite{GuruswamiLJ20} &  $k \le d < n$ & $n-o(n)$ \\ 
 \hline 
 This paper & $k \le d < n$ & None \\   
 \hline
\end{tabular} 
 \caption{Comparison between different constructions.}
 \label{table:compare}
 \end{center}
 \end{table}
\subsection{Multiple Node Failures and Multiple Repair Degrees}
Suppose \(h\le r\) nodes failed in an \((n,k=n-r,\ell)_\mathbb{F}\) MDS array code, and suppose \(d\) helper nodes are contacted to repair it. We assume a centralized repair setting where a central entity collects the helper data and repairs the failed nodes. In this case it is known \cite{CadJafMalRamSuh} the each helper node needs to transmit at least \(\frac{h\ell}{d-k+h}\) symbols (over \(\mathbb{F}\)). Following the terminology from \cite{YeBar_1}, we say that an \((n,k,\ell)_\mathbb{F}\) MDS array code has the \((h,d)\) optimal repair property if it is possible to repair the contents of any \(h\) failed nodes by downloading \(\frac{h\ell}{d-k+h}\) symbols from each of any \(d\) surviving nodes. It follows from this definition that an \((n,k,d,\ell)_\mathbb{F}\) MSR code has the \((1,d)\) optimal repair property. It should also be noted that the MDS property is equivalent to \((h=r,d=k)\) optimal repair property.

The number of failed nodes and available nodes can change over time. Therefore, the ability to simultaneously support all possible \(h\le r\) and \(d\le n-h\) is a useful property. In \cite{YeBar_1}, two explicit families of MDS array codes having \((h,d)\) optimal repair property for all \( (h,d) \) such that \(1 \le h \le r\) and \(k \le d \le n-h\) simultaneously are presented. 
The codes in \cite{YeBar_1} require a sub-packetization level \(\ell = \theta^n\), where \(\theta = \text{lcm}(1,2,\dots,r)\). 
In \cite{zhang2024}, such MDS array codes with \(\ell = \theta r^n\) are provided. Clearly, the sub-packetization levels of these codes are exponential in \(n\) for constant \(r\). Hence, reducing the sub-packetization level by slightly sacrificing the repair bandwidth is meaningful for these types of codes as well.

We say that an \((n,k,\ell)_\mathbb{F}\) MDS array code has \((h,d)\) \(\varepsilon\)-optimal repair property if the contents of any \(h\) failed nodes can be recovered by downloading \((1+\varepsilon)\frac{h\ell}{d-k+h}\) symbols from each of any \(d\) helper nodes. Note that setting \(\varepsilon = 0\) yields the \((h,d)\) optimal repair property. It can be seen that an \((n,k,\ell)_\mathbb{F}\) \(\varepsilon\)-MSR code has \((1,d)\) \(\varepsilon\)-optimal repair property. An \((n,k,\ell)_\mathbb{F}\) MDS array code with \((h,d)\) \(\varepsilon\)-optimal repair property is said to have \((h,d)\) \(\varepsilon\)-optimal help-by-transfer property if the symbols downloaded from each helper node during repair are a subset of the symbols stored in that helper node.
	
	\subsection{Our Contribution}
	In this paper, we present the first construction of \(\varepsilon\)-MSR codes that allows all repair degrees, i.e., \(k 
 \le d < n\), over a field of size \(O(n(d-k))\). Moreover, this is also the first construction of \(\varepsilon\)-MSR codes with the help-by-transfer property. 
More precisely, given an \(\varepsilon > 0\) and fixed \(k < d\), we can construct an \((n,k,d,\ell)_\mathbb{F}\) \(\varepsilon\)-MSR code whose sub-packetization \(\ell\) scales logarithmically with length \(n\). Our construction  
utilizes two building blocks. The first is a new MDS array code construction, and the second is a code with good parameters, i.e., large size and distance. These two building blocks are combined together, as in \cite{RawTamGur_epsilonMSR1}, to yield the \(\varepsilon\)-MSR code construction. 
This construction is then generalized to obtain \((n,k=n-r,\ell)\) MDS array codes with the \((h,d)\) \(\varepsilon\)-optimal help-by-transfer property for all \((h,d)\) with \(h \le r\) and \(k \le d \le n-h\) simultaneously. For this code as well, the sub-packetization \(\ell\) scales logarithmically with length \(n\). 
 
	\section{ Construction of $\varepsilon$-MSR Codes} \label{sec:eps_msr}
We first define a group algebra which is needed to describe our construction. 
	\subsection{A Group Algebra} \label{sec:group_algebra}

 For a positive integer \(s\), let \(\mathbb{Z}_s = \{0, \ldots, s-1\}\) be the additive group of integers modulo \(s\), and let \(G = \mathbb{Z}_s^t\) be the abelian group formed by the \(t\)-fold direct product of \(\mathbb{Z}_s\).
Let \(\mathbb{F}\) be a finite field, and define the group algebra of \(G\) over \(\mathbb{F}\) by
\[
\mathbb{F}[G] = \left\{ \sum_{g \in G} a_g x_g \mid a_g \in \mathbb{F} \right\},
\]
which consists of all formal sums of \(x_g\), \(g \in G\), with coefficients from \(\mathbb{F}\). It is easy to see that \(\mathbb{F}[G]\) is a vector space over \(\mathbb{F}\) of dimension \(|G| = s^t\), spanned by the standard basis \(\{x_g \mid g \in G\}\). Addition in the algebra is defined coordinate-wise, i.e.,
\[
\sum_{g \in G} a_g x_g + \sum_{g \in G} b_g x_g := \sum_{g \in G} (a_g + b_g) x_g.
\]
Multiplication of two standard basis vectors \(x_v, x_u\) is defined as \(x_v \cdot x_u = x_{v+u}\), and this is then extended linearly to all of \(\mathbb{F}[G]\).
Note that since $G$ is a commutative group, it follows that the multiplication in the algebra is also commutative. 

Given an element \(f \in \mathbb{F}[G]\), one can define a linear transformation on the vector space \(\mathbb{F}[G]\) simply by multiplying by \(f\), i.e., \(h \mapsto h \cdot f\) for any \(h \in \mathbb{F}[G]\). This linear transformation can be represented via a matrix after selecting a basis, which we set to be the standard basis. Thus, each element \(f \in \mathbb{F}[G]\) can be viewed as an \(s^t \times s^t\) matrix over \(\mathbb{F}\).
In fact, the mapping that sends \(f \in \mathbb{F}[G]\) to its matrix representation is an algebra homomorphism between the algebra \(\mathbb{F}[G]\) and the algebra of matrices with the usual addition and multiplication, which is called the regular representation. Note that in particular, the matrix representation of $x_v\in \mathbb{F}[G]$ for $v\in G$ is an $s^t\times s^t$  permutation matrix. We proceed to give a simple example of a group algebra and a matrix representation of its elements.

	\begin{example}
	    Let \( G = \mathbb{Z}_2^2 = \{(0,0), (0,1), (1,0), (1,1)\} \), i.e., \( s = t = 2 \). Let \(\mathbb{F}\) be the binary field \(\mathbb{F}_2\). Then, \(\mathbb{F}[G]\) is a vector space over \(\mathbb{F}_2\) of dimension 4. Let \(\rho\) be the mapping that takes an element of the group algebra and maps it to its matrix representation under the standard basis, where we fix the following order of the standard basis: \( x_{(0,0)}, x_{(0,1)}, x_{(1,0)}, x_{(1,1)} \).

One can verify that
\begin{align*}
    \rho(x_{(0,0)}) &= \begin{pmatrix}
        1 & 0 & 0 & 0 \\
        0 & 1 & 0 & 0 \\
        0 & 0 & 1 & 0 \\
        0 & 0 & 0 & 1 
    \end{pmatrix},~ 
    \rho(x_{(0,1)}) = \begin{pmatrix}
        0 & 1 & 0 & 0 \\
        1 & 0 & 0 & 0 \\
        0 & 0 & 0 & 1 \\
        0 & 0 & 1 & 0 
    \end{pmatrix},~ 
    \rho(x_{(1,0)}) = \begin{pmatrix}
        0 & 0 & 1 & 0 \\
        0 & 0 & 0 & 1 \\
        1 & 0 & 0 & 0 \\
        0 & 1 & 0 & 0 
    \end{pmatrix},~ 
    \rho(x_{(1,1)}) = \begin{pmatrix}
        0 & 0 & 0 & 1 \\
        0 & 0 & 1 & 0 \\
        0 & 1 & 0 & 0 \\
        1 & 0 & 0 & 0 
    \end{pmatrix}.
\end{align*}

It can be verified that for all \( u, v \in G \), we have \(\rho(x_v \cdot x_u) = \rho(x_u) \cdot \rho(x_v)\) and 
\(\rho(x_v + x_u) = \rho(x_u) + \rho(x_v)\), i.e., \(\rho\) is a ring homomorphism. For instance,
\begin{align*}
    \rho(x_{(0,1)}) \cdot \rho(x_{(1,0)}) &= \begin{pmatrix}
        0 & 1 & 0 & 0 \\
        1 & 0 & 0 & 0 \\
        0 & 0 & 0 & 1 \\
        0 & 0 & 1 & 0 
    \end{pmatrix} \cdot \begin{pmatrix}
        0 & 0 & 1 & 0 \\
        0 & 0 & 0 & 1 \\
        1 & 0 & 0 & 0 \\
        0 & 1 & 0 & 0 
    \end{pmatrix} = \begin{pmatrix}
        0 & 0 & 0 & 1 \\
        0 & 0 & 1 & 0 \\
        0 & 1 & 0 & 0 \\
        1 & 0 & 0 & 0 
    \end{pmatrix} = \rho(x_{(1,1)}) = \rho(x_{(0,1)} \cdot x_{(1,0)}).
\end{align*}

Similarly,
\begin{align*}
    \rho(x_{(0,0)}) + \rho(x_{(0,1)}) &= \begin{pmatrix}
        1 & 0 & 0 & 0 \\
        0 & 1 & 0 & 0 \\
        0 & 0 & 1 & 0 \\
        0 & 0 & 0 & 1 
    \end{pmatrix} + \begin{pmatrix}
        0 & 1 & 0 & 0 \\
        1 & 0 & 0 & 0 \\
        0 & 0 & 0 & 1 \\
        0 & 0 & 1 & 0 
    \end{pmatrix} = \begin{pmatrix}
        1 & 1 & 0 & 0 \\
        1 & 1 & 0 & 0 \\
        0 & 0 & 1 & 1 \\
        0 & 0 & 1 & 1 
    \end{pmatrix} = \rho(x_{(0,0)} + x_{(0,1)}).
\end{align*}

	\end{example}
In the sequel, we will omit the notation $\rho$ and view each element of the group algebra as a matrix. 

Since \(\mathbb{F}[G]\) is a commutative ring, the following two properties hold:

\begin{enumerate}
    \item \(\mathbb{F}[G][X]\), which consists of all polynomials in the variable \(X\) with coefficients from \(\mathbb{F}[G]\), is also a commutative ring.
    \item For any \(a \in \mathbb{F}[G]\) and \(g(X) \in \mathbb{F}[G][X]\), the mapping 
   \[
   g(X) \mapsto g(a)
   \]
   is a ring homomorphism from \(\mathbb{F}[G][X]\) to \(\mathbb{F}[G]\) (called the evaluation homomorphism). Therefore, for any \(g(X), f(X) \in \mathbb{F}[G][X]\),
   \[
   (g(X) \cdot f(X))(a) = g(a) \cdot f(a).
   \]
\end{enumerate}

We proceed to give the construction of an MDS array code, which will serve as a building block in the construction of the $\varepsilon$-MSR code.
	\subsection{An MDS Array Code Construction} \label{sec:MDS_code}
For positive integers \(k < n \), \( s\leq n-k\) and \(t \le n\), we present an \((n,k,\ell=s^t)_{\mathbb{F}}\) MDS array code construction, where any failed node can be repaired by downloading information from any \(d=k+s-1\) helper nodes. Note that indeed $k\leq d<n$. The guarantees on the repair bandwidth are given in Theorem \ref{sec:MDS} below. The code is defined over the finite field \(\mathbb{F}\) of order \(q \ge gn + 1\), where \(g = \gcd(d - k + 1, q - 1)\). Since \(g \le d - k + 1\), it follows that \(q = O(n(d - k))\).
 

Set $r=n-k$ and let $\alpha$ be a primitive element of $\mathbb{F}$. 
 For $i\in[n]$ let~$\alpha_i\triangleq \alpha^{i-1}$,  and for~$i\in[t]$ let~$e_i$ be the $i$-th standard basis vector of $G$,  i.e., the all zeros vector of length $t$ except for the $i$-th entry which is $1.$
 
 We now present the code construction.  
 
Recall that we consider any element of~$\mathbb{F}[G]$ as an $s^t\times s^t$ square matrix over~$\mathbb{F}$. In particular, for $i \in [t]$, $x_{e_i}\in \mathbb{F}[G]$ is an $s^t\times s^t$  permutation matrix. 
\begin{construction}\label{construction:MDS}
    Fix~$(a_1,\ldots,a_n)\in\{e_1,\ldots,e_t\}^n$.  For~$i\in[r]$ and~$j\in[n]$, let
    $$A_{i,j} \coloneqq (\alpha_j\cdot x_{a_j})^{i-1}\in\mathbb{F}^{s^t\times s^t}.$$ The code's parity check matrix is $$A \coloneqq (A_{i,j})_{i \in [r],j\in [n]}\in \bF^{rs^t\times ns^t}.$$	
	\end{construction}

 Note that  Construction~\ref{construction:MDS} gives a family of codes, with each code corresponding to a specific choice of $\{a_j\}_{j \in [n]}$.  In the discussion below, we focus on an arbitrary code from this family.
 The following theorem summarizes the properties of the construction. 
 \begin{theorem}
 \label{sec:MDS}
     For every $a_1,\dots,a_n$, the code constructed in Construction \ref{construction:MDS} is an \((n,k=n-r,\ell=s^t)_\mathbb{F}\) MDS array code with the following repair property.  For any failed node $i\in [n]$ and a subset $D\subseteq [n]\backslash\{i\}$ of $d=k+s-1$ helper nodes, one can repair node $i$ by downloading from node $j\in D$ a fraction of $f_j$ of it data, where 
     $$f_j=\begin{cases}\frac{1}{s} & a_j\neq a_i\\ 
     1 & \text{else.}
     \end{cases}$$
     Moreover,
     the help-by-transfer property holds.
     \end{theorem}

 \begin{proof}
	    {\bf MDS Property:} 
		It suffices to show that any $r \times r$ block submatrix of $A$ is invertible. 
  Without loss of generality, it is shown that the leftmost $r\times r$ block submatrix, i.e., 
		\begin{align}
		    \begin{pmatrix}
			I &\cdots &I\\
			\alpha_1x_{a_1} & \cdots& \alpha_{n-k} x_{a_{n-k}}\\
			\vdots &\ddots &\vdots\\
			(\alpha_1x_{a_1})^{n-k-1} & \cdots& (\alpha_{n-k} x_{a_{n-k}})^{n-k-1}
		\end{pmatrix} 
  \label{eq:MDS}
		\end{align}
		is invertible. Since all the block matrices in the above matrix commute, it follows from \cite[Equation 1]{BenYaacov2014} and \cite[Thoerem 1]{Silvester2000} 
  that the determinant of \eqref{eq:MDS} equals 
		$$\prod_{1\leq i <j\leq r}\det (\alpha_ix_{a_i}- \alpha_jx_{a_j}).$$
		If $a_i=a_j$ for some $i,j \in [r]$, then clearly the matrix $\alpha_ix_{a_i}- \alpha_jx_{a_j}$ is invertible as by our selection $\alpha_i\neq \alpha_j$ and since $x_{a_i}$ is a permutation matrix. If $a_i\neq a_j$, then it suffices to show that  the matrix 
		$I-\alpha_j\alpha_i^{-1}x_{a_j}x_{a_i}^{-1}$ is invertible. It is easy to verify that the order of the matrix $x:=x_{a_j}x_{a_i}^{-1}$ is $s$, i.e., $s$ is the smallest positive integer $u$  for which  $x^u=I.$ 
		Set $\beta \coloneqq \alpha_j\alpha_i^{-1}$, then by Lemma~\ref{lem:alpha} in Appendix~\ref{appendix:alpha}, we have $\beta^s\neq 1$.  Assume towards a contradiction that there exists a nonzero vector $v$ for which $(I-\beta x)v=0$, equivalently, $\beta xv=v$, and by induction also $(\beta x)^sv=v$. However, $(\beta x)^sv=\beta^sx^sv=\beta^sIv=\beta^sv\neq v$, and we arrive at a contradiction. 
	
	\vspace{0.2cm}
	{\bf Repair  property:} 
				It can be seen that for $p \in [r]$, the $p$-th block row of the parity check matrix $A$ of the code is obtained by evaluating the polynomial $X^{p-1}\in \mathbb{F}[G][X]$ at the ring elements $\alpha_jx_{a_j}, j \in [n].$	
		Let $\overline{D}=[n]\backslash (\{i\}\cup D)$ be the set of non-helper nodes, which is of size $n-d-1$.  Let $h(X)\in \mathbb{F}[G][X]$ be the annihilator polynomial of the set $\overline{D}$, i.e., 
		$$h(X)=\prod_{j\in \overline{D}}(X-\alpha_jx_{a_j}).$$ 
		If $d=n-1$, then $\overline{D}$ is the  empty set, and we define $h(X)=I$.  		
		For $m=0,\ldots,s-1$, let $f_m(X)=X^mh(X)$. Clearly $\deg(f_m)\leq s-1+n-d-1=r-1.$ Hence, the  vector of matrices 
		$$(f_m(\alpha_1x_{a_1}),\ldots,f_m(\alpha_nx_{a_n}))$$
		is in the row span of $A$. 
		To see this, write  $f_m(X)=\sum_{j=0}^{r-1}b_jX^j$, where $b_j\in \mathbb{F}[G]$. 
 Referring to $A$ as an $r \times n$ matrix over $\mathbb{F}[G]$, we have
  $$(b_0,\ldots,b_{r-1})\cdot A=(\textstyle\sum_{j=0}^{r-1}b_j(\alpha_1x_{a_1})^j,\ldots,\sum_{j=0}^{n-k-1}b_j(\alpha_nx_{a_n})^j)=(f_m(\alpha_1x_{a_1}),\ldots,f_m(\alpha_nx_{a_n})).$$
		Since evaluation is a ring homomorphism, we have that 
$$(f_m(\alpha_1a_1),\ldots,f_m(\alpha_na_n))=((\alpha_1x_{a_1})^mh(\alpha_1x_{a_1}),\ldots,(\alpha_nx_{a_n})^mh(\alpha_nx_{a_n})).$$
In addition, since $h$ is the annihilator polynomial of $\overline{D}$, it follows that for any $j\in \overline{D}$, the $j$-th entry (which is a matrix) is the zero matrix.
 For a codeword $\boldc=(\boldc_1,\ldots,\boldc_n)$, we have $A(\boldc_1,\ldots,\boldc_n)^\intercal=\mathbf{0}$, and hence
		\begin{equation}
			\label{eq1}
			((\alpha_1x_{a_1})^mh(\alpha_1x_{a_1}),\ldots,(\alpha_nx_{a_n})^mh(\alpha_nx_{a_n}))\cdot (\boldc_1,\ldots,\boldc_n)^\intercal 
			= 0~~\text{ for }~m\in\{0,1,\ldots,s-1\}. 
		\end{equation}
Recall that ~$(a_1,\ldots,a_n)\in\{e_1,\ldots,e_t\}^n$. Assume without loss of generality that $a_i=e_1$, where $i\in[n]$ is the node to be repaired. 
Let $$S \coloneqq \text{span}\{x_g \vert g(1)=0\}$$
  be the subspace of $\mathbb{F}[G]$ spanned by the basis vectors $x_g$ for $g$'s whose first coordinate equals zero. Clearly, $\dim(S)=s^{t-1}$.
  		By abuse of notation, we also view \(S\) as a \(0,1\) matrix over $\mathbb{F}_q$ of order \(s^{t-1} \times s^t\) whose row span equals the subspace \(S\). Indeed, let the \(s^t\) columns of \(S\) be indexed by the elements of \(G\), and let each row of \(S\) correspond to one element in the basis of \(S\). More precisely, each row is a zero vector except for one coordinate which is one, corresponding to some element \(x_g\) with \(g(1) = 0\).
 
 For any $j\neq 1$, the subspace $S$ is an invariant subspace under the multiplication by  $x_{e_j}$, i.e., $$Sx_{e_j}=S$$ 
 Also, the $s$ subspaces $S,Sx_{e_1},\ldots,Sx_{e_1}^{s-1}$ form a direct sum of the vector space 
		$\mathbb{F}[G]$, i.e., 
		\begin{equation}
			\label{eq:direct_sum}
			S\oplus Sx_{e_1}\oplus\ldots\oplus Sx_{e_1}^{s-1}=\mathbb{F}[G].
		\end{equation}
The two statements given above about $S$ can be proved as follows. Let $G'$ be the subgroup of $G$ formed by elements whose first coordinate is zero. Since it contains $e_j$ for all $j \ne 1$, it follows that $G'+e_j=G'$, which implies the first statement. Since $s$ cosets of $G'$ are $G', G'+e_1, G'+2e_1,...,G'+(s-1)e_1$, the second statement follows.

		We claim that it is possible to recover the vector $\boldc_i$ using the $s$ equations in \eqref{eq1}. Indeed, by \eqref{eq1} we have 
	$$(\alpha_ix_{a_i})^mh(\alpha_ix_{a_i})\boldc_i^\intercal=-\sum_{j\neq i}(\alpha_jx_{a_j})^mh(\alpha_jx_{a_j})\boldc_j^\intercal ~~\text{ for }~m\in\{0,1,\ldots,s-1\}.$$
  By multiplying from the left  by $S$ we get
$$S(\alpha_ix_{a_i})^mh(\alpha_ix_{a_i})\boldc_i^\intercal=-\sum_{j\neq i}S(\alpha_jx_{a_j})^mh(\alpha_jx_{a_j})\boldc_j^\intercal~~\text{ for }~m\in\{0,1,\ldots,s-1\}$$
Since $h(\alpha_jx_{a_j})=0$ for $j\in \overline{D}$, we have 
\begin{align}
 S(\alpha_ix_{a_i})^mh(\alpha_ix_{a_i})\boldc_i^\intercal=
		-\sum_{j\in D}S(\alpha_jx_{a_j})^mh(\alpha_jx_{a_j})\boldc_j^\intercal~~\text{ for }~m\in\{0,1,\ldots,s-1\}. 
  \label{eq:recovery}
\end{align}
		We now show that if we get enough information from the nodes in $D$ such that we can construct the RHS of the above $s$ equations, then we can recover the lost vector $\boldc_i$. 
  Since $a_i=e_1$ and $\alpha_i \in \mathbb{F} \setminus \{0\} $, it follows from \eqref{eq:direct_sum} that
  $$S \oplus S(\alpha_ix_{a_i})\oplus\ldots\oplus S(\alpha_ix_{a_i})^{s-1}=\mathbb{F}[G].$$
  One can verify that $h(\alpha_ix_{a_i})$ is an   invertible matrix; this is true since $\alpha_ix_{a_i}-\alpha_jx_{a_j}$ is invertible whenever $i \ne j$ (follows from the MDS property proof) and $h(\alpha_ix_{a_i})$ is a product of such matrices.
  Therefore, we also have that
\begin{align}
    Sh(\alpha_ix_{a_i})\oplus S(\alpha_ix_{a_i})h(\alpha_ix_{a_i})\oplus\ldots\oplus S(\alpha_ix_{a_i})^{s-1}h(\alpha_ix_{a_i})=\mathbb{F}[G],  
    \label{eq:direct_sum_final}  
\end{align}
which implies that we can recover node $i$ using the $s$ equations. 
This can be argued as follows. It follows from \eqref{eq:direct_sum_final} that the $s$ many $s^{t-1} \times s^t$ matrices which multiply $\boldc_i^\intercal$ from the left in the LHS of \eqref{eq:recovery}, when stacked on top of each other, form an invertible matrix. Hence, one can use \eqref{eq:recovery} to compose an invertible linear system whose solution $\boldc_i$ can be easily computed.  
		
		Next, we calculate the amount of information needed from the nodes in $D$ to compute the RHS of $s$ equations in \eqref{eq:recovery}. Fix $j\in D$, then for the $s$ equations to be constructed, we will need the following information from node $j$:
		$$S(\alpha_jx_{a_j})^mh(\alpha_jx_{a_j})\boldc_j^\intercal \text{ for }~ m\in\{0,1,\ldots,s-1\}.$$
		However, as noted earlier, if $a_j\neq a_i$ then $S$ is an invariant subspace of the operator $x_{a_j}$, which means that only $Sh(\alpha_jx_{a_j})\boldc_j^\intercal$ is needed from helper node $j$.  Furthermore, the subspace $S$ is also an invariant subspace under the multiplication by $h(\alpha_jx_{a_j})$, as expansion of $h(\alpha_jx_{a_j})$ does not contain a term with $x_{a_i}$. Therefore, helper node $j$ with $a_j\neq a_i$ only needs to send  $S\boldc_j^\intercal$, which is exactly $1/s$ fraction of the information it stores. 
		Lastly, if $a_j=a_i$, then by \eqref{eq:direct_sum_final}, we see that all the information stored in node $j$ is needed. Clearly, the help-by-transfer property follows since all rows of $S$ are unit vectors.
	\end{proof} 	
	\begin{remark}
		\normalfont In Construction~\ref{construction:MDS}, if we set \(t = n\) and \(a_j = e_j\) for all \(j \in [n]\), then the resulting code is an \((n,k,d,\ell = s^n)_{\mathbb{F}}\) MSR code with the help-by-transfer property. We remark that this code has some structural similarity with an MSR code from \cite{YeBar_1}, which uses generalized permutation matrices to construct the parity-check matrix.
	\end{remark}
	\subsection{Transformation to $\varepsilon$-MSR Codes}
	\label{sec:eps_msr_transformation}
In the previous section, we constructed an \((n,k,\ell=s^t)\) MDS array code, which will be the first ingredient in the construction of an \(\varepsilon\)-MSR code. The second ingredient is a large code with a large normalized minimum distance. More precisely, given a code \(\mathcal{U} = \{\boldsymbol{u}^{(1)}, \ldots, \boldsymbol{u}^{(n)}\}\) over the alphabet \([t]\) with block length \(\lambda\), size \(n\), and normalized minimum distance at least \(\delta\), we construct an \((n,k,d,\ell=\lambda s^t)_{\mathbb{F}}\) \(\varepsilon\)-MSR code with $\epsilon=(1-\delta)(s-1)$ as follows. 
 We use the notation \(\boldsymbol{u}^{(j)} = (u^{(j)}_1, \dots, u^{(j)}_{\lambda})\) to represent the symbols of the \(j\)-th codeword.

\begin{construction} \label{construction:eps_MSR} 
For each codeword \(\boldsymbol{u}^{(j)}\in\mathcal{U}\) define a block diagonal matrix \(H_j\in\mathbb{F}^{\lambda s^t\times\lambda s^t}\), with~$\lambda$ blocks of size~$s^t\times s^t$, where the~$b$-th block equals~$\alpha_j x_{e_{m}}$ with~$m=u_b^{(j)}$, and with~$\alpha_j=\alpha^{j-1}$ as in Construction \ref{construction:MDS}. The code's parity check matrix is $$H = (H_j^{i-1})_{i \in [r], j \in [n]}\in \mathbb{F}^{r \lambda s^t \times n \lambda s^t}.$$
\end{construction}
 See Figure~\ref{fig:example} for an example of Construction~\ref{construction:eps_MSR}. 
Each node stores a vector over $\mathbb{F}$ of length $\lambda s^t$, and each such vector can be partitioned into $\lambda$ chunks of length $s^t$. More formally, 
the vector stored on node $j$ can be written as $(\boldc_{j,1}, \boldc_{j,2}, \dots, \boldc_{j,\lambda}) \in \mathbb{F}^{\lambda s^t}$ 
with each chunk $\boldc_{j,b} \in \mathbb{F}^{s^t}$. 
It follows from the structure of $H$ that  
the following observation holds. 
\begin{observation} \label{obsv:independent_encodings}
  For every $b \in [\lambda]$, $(\boldc_{1,b}, \boldc_{2,b}, \dots, \boldc_{n,b})$ is a codeword of the code obtained by setting by $a_j=e_{u^{(j)}_b}$ for all $j \in [n]$ in Construction~\ref{construction:MDS}. Thus we have $\lambda$ independent encodings, with each encoding corresponding to some particular choice of $(a_1,\ldots,a_n)\in\{e_1,\ldots,e_t\}^n$ in Construction~\ref{construction:MDS}.
\end{observation}
 It follows from Theorem~\ref{sec:MDS} that each of these $\lambda$ encodings has MDS property. Therefore, the code with $H$ as parity check matrix is an $(n,k,\ell=\lambda s^t)_{\mathbb{F}}$ MDS array code. 
 		\begin{figure*}[ht!]
		\begin{center}
			\includegraphics[scale=0.5]{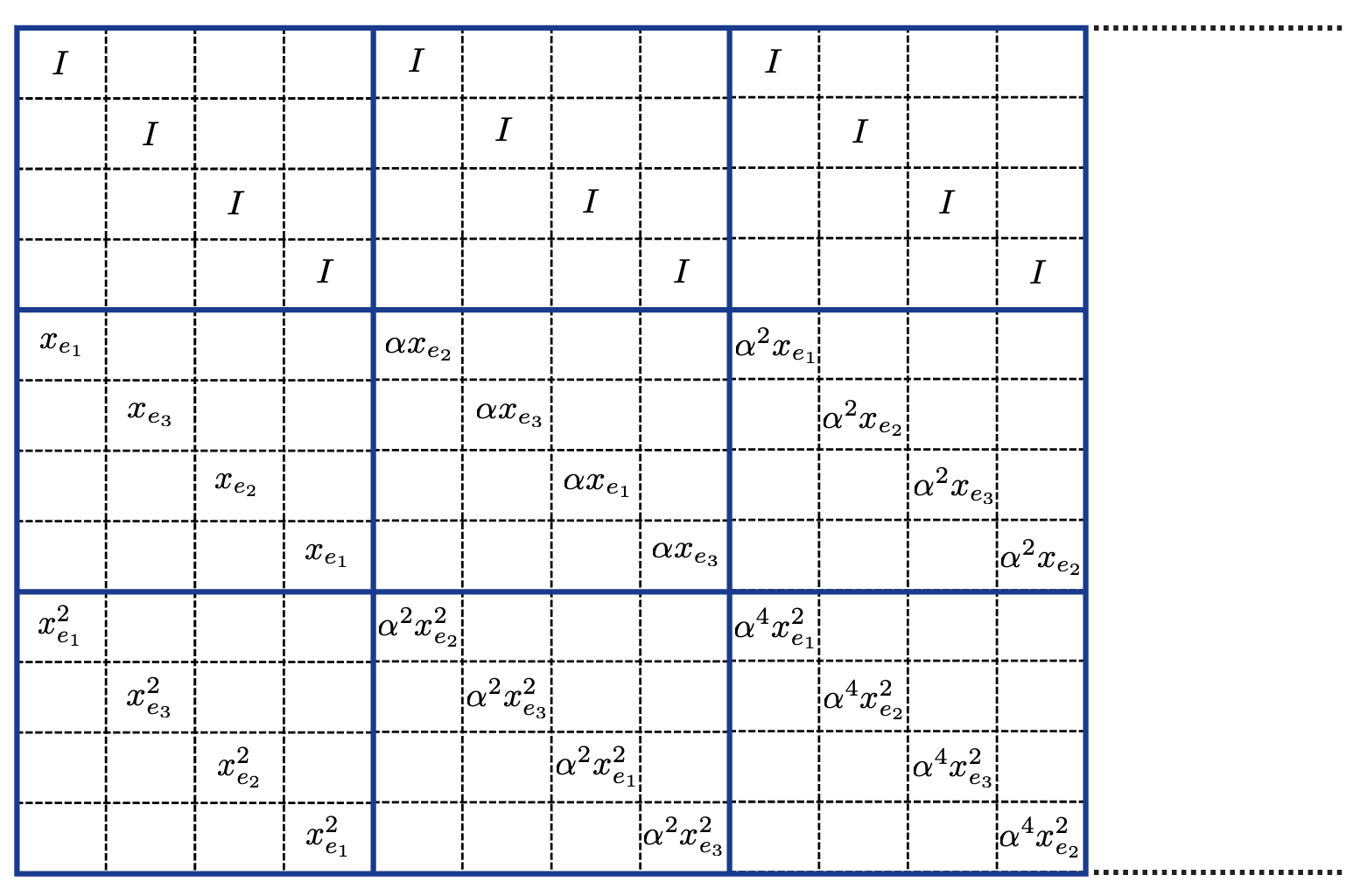}
			\caption{An illustration of the parity check matrix given by Construction~\ref{construction:eps_MSR}, for $r=3$, $\lambda=4$ and $t=3$, where $\boldu^{(1)}=(1,3,2,1)$, $\boldu^{(2)}=(2,3,1,3)$ and $u^{(3})=(1,2,3,2)$ are three codewords of $\cU$. Only the columns corresponding to the three nodes indexed by these three codewords are depicted in the figure.}
			\label{fig:example}
		\end{center}
	\end{figure*} 
 

We now show that Construction~\ref{construction:eps_MSR} yields an $(n,k,d, \ell)_\mathbb{F}$ $\varepsilon$-MSR code with $\varepsilon=(1-\delta)(s-1)$.
	\begin{lemma} 
 \label{lem:eps}
 Let  $\varepsilon=(1-\delta)(s-1)$.
	    For the code given by Construction~\ref{construction:eps_MSR}, any failed node can be repaired by downloading at most $(1+\varepsilon)\frac{\ell}{s}=(1+\varepsilon)\frac{\ell}{d-k+1}$ symbols over $\mathbb{F}$ from each of any $d$ chosen helper nodes. 
	\end{lemma}
 \begin{proof} 
 Assume node $i$ needs to be repaired using the $d$ helper nodes indexed by $D \subseteq [n]\backslash \{i\}$. The $s^t$ symbols of node $i$ corresponding to each of the $\lambda$ encodings (see Observation~\ref{obsv:independent_encodings}) can be repaired independently. By the repair property in Theorem \ref{sec:MDS}, for the repair of the $s^t$ symbols of node $i$ which are involved in $b$-th encoding, node $j \in D$ transmits $s^{t-1}$ symbols if $u^{(j)}_b \ne u^{(i)}_b$ and otherwise it transmits all the $s^t$ symbols corresponding to that encoding. 
	In either case, node $j$ accesses only the symbols it transmits. 
	If $j \notin D$, then node $j$ does not send any information. 
 
 Since the minimum distance of $\cU$ is  at least $ \delta \lambda$, for every $j\ne i$ we have $u^{(j)}_b \ne u^{(i)}_b$ for at least $\delta \lambda$ values of $b \in [\lambda]$. Equivalently, in at most $(1-\delta)\lambda$ encodings, the entire information ($s^t$ symbols) is sent by the helper node $j$. If $j \notin D$, then node $j$ does not participate in the repair and does not send any information.
	In total, the download from each helper node is at most 
	$$\delta \lambda s^{t-1}+(1-\delta)\lambda s^t= \lambda s^{t-1}+ (1-\delta)\lambda(s^t- s^{t-1})=\frac{\ell}{s}(1+(1-\delta)(s-1)).$$
	Hence, for $\varepsilon=(1-\delta)(s-1)$, the code given by Construction~\ref{construction:eps_MSR} is an $(n,k,d=k+s-1,\ell=\lambda s^t)_{\mathbb{F}}$ $\varepsilon$-MSR code with help-by-transfer property. 
 \end{proof}
The next theorem easily follows from Lemma~\ref{lem:eps}.
\begin{theorem}
    The code defined in Construction \ref{construction:eps_MSR} is an $(n,k,d=k+s-1,\ell=\lambda s^t)_\mathbb{F}$ $\varepsilon$-MSR code with $\varepsilon=(1-\delta)(s-1)$.
\end{theorem}

 \subsection{Existence of  $\varepsilon$-MSR Codes with Logarithmic Sub-packetization Level} \label{sec:existence}
	Here we show that our construction leads to an  $(n,k=n-r,d=k+s-1,\ell)_{\mathbb{F}}$ $\varepsilon$-MSR codes with sub-packetization level $\ell$ scaling logarithmically with $n$, for any fixed  $k< d $. Equivalently, for any fixed $k<d$  the length of the code $n$ scales exponentially with the sub-packetization level $\ell.$

 Given $\varepsilon>0$, let $\delta=1-\frac{\varepsilon}{d-k}=1-\frac{\varepsilon}{s-1}$ and choose a positive integer $t>\frac{s-1}{\varepsilon}$. 
	Then, by  the Gilbert-Varshamov (GV) bound \cite{ecc_Mac_Slo}  there exists a code over the alphabet $[t$] with block length $\lambda$, relative minimum distance at least $\delta$ and code size  at least $t^{\lambda(1-h_t(\delta)-o(1))}$, where 
   $h_t(x)=x\log_t(t-1)-x\log_tx-(1-x)\log_t(1-x)$ is the $t$-ary entropy function. 
	Let  $\cU$  be such a code, then  since~$\ell=\lambda s^t$, it follows that Construction \ref{construction:eps_MSR} yields an $\varepsilon$-MSR code with length $n= t^{((1-h_t(\delta)-o(1))/s^t)\ell}$. For constant $s$ and $t$, this gives $n=\Omega(\exp(\theta \ell))$ for some constant $\theta>0$. Such a code can be constructed over any finite field of size greater than $sn$. Thus, we have the following result. 
	\begin{theorem} \label{thm:eps_msr}
		Given positive integers $r \ge s $ and  $\varepsilon >0$, there is a constant $\theta=\theta(s,\varepsilon)$ such that for infinite values of  $\ell$ there exist  $(n=\Omega(\exp(\theta\ell)),k=n-r,d=k+s-1,\ell)_\mathbb{F}$ $\varepsilon$-MSR codes with help-by-transfer property. Moreover, the required field size $|\mathbb{F}|$ scales linearly with $n$, for constant $s$ . 
	\end{theorem}
	Explicit constructions of codes achieving the GV bound are not known in general. To obtain explicit constructions of $\varepsilon$-MSR codes, one can choose $\cU$ to be the best-known explicit codes. In particular, for a special case, explicit constructions of codes that surpass the GV bound are known. 
\begin{remark} \normalfont \label{remark:AG}
	If $t \ge 49$ is the square of a prime power,  then explicit constructions of Algebraic Geometry (AG) codes over the alphabet $[t]$ beating the GV bound are known \cite{garcia1995tower,garcia1996asymptotic,ShumAKSD01}. 
	For block lenght $\lambda$ and relative minimum distance $\delta$, AG code of size $ t^{\lambda\big(1-\delta-\frac{1}{1-\sqrt{t}}-o(1)\big)}$ can be explicitly constructed. If we choose $\cU$ to be this AG code, then we get an $\varepsilon$-MSR code with $n= t^{(1-\delta-\frac{1}{1-\sqrt{t}}-o(1))/s^t)\ell}$. 
	For constant $s$ and $t$, this gives $n=\Omega(\exp(\theta' \ell))$, for some constant $\theta'$. If $t$ satisfies the above requirements, then $\theta'$ is larger 
 than the constant $\theta$ given by the GV bound.   
\end{remark}
	\section{Construction of MDS Array Codes with $(h,d)$ $\varepsilon$-optimal repair property property} \label{sec:simul}
	In this section, we generalize our $\varepsilon$-MSR construction to obtain $(n,k=n-r,\ell)_{\mathbb{F}}$ MDS array codes with $(h,d)$ $\varepsilon$-optimal repair property simultaneously for all $1 \le h \le r$ and $k \le d \le n-h$. 
For~$\zeta = \operatorname{lcm}(1,2,\ldots,r)$, let $G=\mathbb{Z}_\zeta^t$, and let~$\mathbb{F}[G]$ be the group algebra of~$G$ over~$\mathbb{F}$.
 Similar to Section~\ref{sec:group_algebra}, $\{x_g \vert g \in G\}$ is a basis of $\mathbb{F}[G]$ over $\mathbb{F}$ and we view each element of $\mathbb{F}[G]$ as an $\zeta^t\times \zeta^t$ matrix over $\mathbb{F}$. 
 Furthermore, suppose that $|\mathbb{F}| \coloneqq q \ge \gcd(\zeta,q-1)n+1.$
 Similar to Section~\ref{sec:MDS_code}, for a primitive element~$\alpha$ of~$\mathbb{F}$ we let~$\alpha_j\triangleq \alpha^{j-1}$, and for~$i\in[t]$ we let~$e_i$ be the~$i$-th standard basis vector of~$G$.
 

	Given a code $\cU=\{\boldu^{(1)},\ldots,\boldu^{(n)}\}$ over $[t]$ of block length $\lambda$, size $n$, and normalized minimum distance $\ge \delta$, we construct an $(n,k,\ell=\lambda \zeta^t)_{\mathbb{F}}$ MDS code as follows.
 We use the notation \(\boldsymbol{u}^{(j)} = (u^{(j)}_1, \dots, u^{(j)}_{\lambda})\) to represent the symbols of the \(j\)-th codeword.
 We note that the only difference between this construction and Construction~\ref{construction:eps_MSR} in the previous section is that $s$ is replaced by $\zeta$. 

\begin{construction} \label{construction:MDS_simul} 
For each codeword \(\boldsymbol{u}^{(j)}\in\mathcal{U}\) define a block diagonal matrix \(H_j\in\mathbb{F}^{\lambda \zeta^t\times\lambda \zeta^t}\), with~$\lambda$ blocks of size~$\zeta^t\times \zeta^t$, where the~$b$-th block equals~$\alpha_j x_{e_{m}}$ with~$m=u_b^{(j)}$, and with~$\alpha_j=\alpha^{j-1}$ as in Construction \ref{construction:MDS}. The code's parity check matrix is $$H = (H_j^{i-1})_{i \in [r], j \in [n]}\in \mathbb{F}^{r \lambda \zeta^t \times n \lambda \zeta^t}.$$
\end{construction}	
 


	
	The content of each node can be partitioned into $\lambda$ chunks of length $\zeta^t$. It follows from the structure of $H$ that the $b$-th chunks from all the $n$ nodes form a codeword of an $(n,k,\ell=\zeta^t)_{\mathbb{F}}$ array code (similar to Observation~\ref{obsv:independent_encodings}). Using arguments similar to those in the proof of Theorem \ref{sec:MDS}, it can be shown that each of these $\lambda$ encodings have MDS property. Therefore, the code with $H$ as a parity check matrix is an $(n,k,\ell=\lambda \zeta^t)_{\mathbb{F}}$ MDS array code. 
	
	\subsection{Repair Property}
	We now prove that the code has the $(h,d)$ $\varepsilon$-optimal repair property, for all $1 \le h \le r$ and all $k \le d \le n-h$ simultaneously. For $(h=r,d=k)$, this follows trivially from the MDS property.
 Fix any arbitrary $h$ and $d$ satisfying $1 \le h \le r-1$ and $k+1 \le d \le n-h$. Pick any $\cF \subseteq [n]$ with $|\cF|=h$ and $D \subseteq [n] \setminus \cF$ with $|D|=d$. 
	Assume that the $h$ nodes indexed by $\cF$ have failed and these need to be repaired using the $d$ helper nodes indexed by $D \subseteq [n]\backslash \cF$. For each failed node, the $\zeta^t$ symbols corresponding to each of the $\lambda$ encodings can be repaired independently of the other symbols. 
	
	\begin{proposition} \label{prop} Fix some $b \in [\lambda]$ and let $a_{j,b} = e_{u^{(j)}_b}$ for all $j \in [n]$. The failed node symbols involved in the $b$-th encoding ($\zeta^t$ symbols per failed node) can be repaired using the following amount of helper information. Every helper node $j \in D$ transmits at most $\frac{h\zeta^{t}}{d-k+h}$ symbols if $a_{j,b} \notin \{a_{i,b} \mid i \in \cF\}$ and otherwise it transmits $\zeta^t$ symbols. In either case, node $j \in D$ accesses only the symbols it transmits. If $j \notin D$, then node $j$ does not send any information. 
	\end{proposition}  
	
	\begin{proof}
		Consider the $r\zeta^t \times n\zeta^t$ matrix  $A^{(b)}=(A_{j}^{i-1})_{i \in [r],j\in [n]}$ over $\bF$ such that $A_{j} \coloneqq \alpha_j\cdot x_{a_{j,b}}$.  Then, the code for the $b$-th encoding is the $(n,k,\ell=\zeta^t)_{\bF}$ MDS array code with $A^{(b)}$ as a parity check matrix. We use $\boldc_{j,b}\in \mathbb{F}^{\zeta^t}$ to denote the information that is stored on node $j$ corresponding to the $b$-th encoding.   
  
Fix $b\in[\lambda]$, and consider the set $W=\{w \in [t] \mid \text{there is at least one}~i \in \cF$~\text{such that}~$u^{(i)}_b=w\}.$	
		Let $z \coloneqq |W|$ and $W =\{w_1,\dots,w_z\}$. Observe that~$z\le h$ since $|\mathcal{F}|=h$.
 For $\mu \in [z]$, set $\cF_\mu = \{i \in \cF \mid  u^{(i)}_b=w_\mu\}.$
Clearly, $\{\cF_\mu\}_{\mu \in [z]}$ is a partition of $\cF$. For $\mu \in [z]$, let $s_{\mu}=d-k+{\mu}$ and $$G_{\mu} = \{g \in G \mid g(w_{\mu})\bmod s_{\mu}=0\},$$ i.e.,~$G_\mu$ is the set of all elements of~$G=\mathbb{Z}_\zeta ^t$ whose $w_\mu$-th entry is divisible by $s_\mu$.

Again for $\mu \in [z]$, we define the subspace $$S_\mu = \text{span} \{x_g \vert g \in G_{\mu}\}$$ of $\mathbb{F}[G]$ with $\dim(S_{\mu})=|G_{\mu}|$. For any $j\neq w_\mu$, the subspace $S_{\mu}$ is an invariant subspace under multiplication by  $x_{e_j}$. The $s_{\mu}$ subspaces $S_{\mu},S_{\mu}x_{e_{w_\mu}},\ldots,S_{\mu}x_{e_{w_\mu}}^{s_{\mu}-1}$ form a direct sum of the vector space 
		$\mathbb{F}[G]$, since the $s_{\mu}$ cosets of $G_{\mu}$ are $G_{\mu}, 
G_{\mu}+e_{w_{\mu}}, \dots, G_{\mu}+(s_{\mu}-1)e_{w_{\mu}}$.  We also view $S_{\mu}$ as a $|G_{\mu}| \times \zeta^t$ matrix over $\mathbb{F}_q$ with $0,1$ entries that represent the basis of $S_{\mu}$, i.e., each row is a zero vector except one coordinate that corresponds to some element $x_g$ with $g \in G_{\mu}$. 
		Additionally, we define subspace $S \coloneqq\text{span}\{x_g \vert g \in \cup_{\mu=1}^{z} G_{\mu}\}$ of dimension $|\cup_{\mu=1}^{z} G_{\mu}|$, which can similarly be viewed as a $|\cup_{\mu=1}^{z} G_{\mu}| \times \zeta^t$ matrix.
		We employ a sequential repair process to obtain the failed symbols $\{\boldc_{i,b} \mid i \in \cF \}$.  This procedure has $z$ steps, of which the $\mu$-th step is as follows. 
	
				\begin{itemize}
				\item \textit{Step $\mu$: }  
					Suppose $\{\boldc_{i,b} \mid i \in \cup_{v=1}^{\mu-1}\cF_{v}\}$ are already repaired (no assumption for $\mu=1$).  
			To repair $\{\boldc_{i,b} \mid i \in \cF_{\mu}\}$,  every node $j \in D$ sends $S_{\mu}\boldc_{j,b}^\intercal$  if $a_j \ne a_{w_\mu}$, 
  and otherwise sends the entire $\boldc_{j,b}$. 
				\end{itemize}
	
		The following lemma shows that identical helper information is needed to repair $\boldc_{i,b}$ for any $i \in \cF_\mu$ and  
		ensures that step $\mu$ repairs $\{\boldc_{i,b} \mid i \in \cF_{\mu}\}$. See Appendix~\ref{appendix:seq_step} for a proof. 
		\begin{lemma}\label{lem:seq_step} For any $\mu \in [z]$ and $i \in \cF_{\mu}$, 
			$\boldc_{i,b}$ can be recovered from $\{\boldc_{j,b} \mid j \in \cup_{v=1}^{\mu-1}\cF_v\} \cup \{S_{\mu}\boldc_{j,b}^\intercal \mid j \in D, a_{j,b} \ne a_{i,b} \} \cup \{\boldc_{j,b}^\intercal \mid j \in D, a_{j,b} = a_{i,b} \}$. If $j \notin D$, then node $j$ does not participate in the repair.  
		\end{lemma}

		Clearly, at the end of the sequential procedure, all the erased contents corresponding to the $b$-th encoding are repaired. If $j \in D$ and $a_{j,b} \in \{a_{i,b} \mid i \in \cF\}$, then entire $\boldc_{j,b}$ is downloaded from node $j$. By definition, $S\boldc_{j,b}^\intercal$ contains all symbols of $S_{\mu}\boldc_{j,b}^\intercal$ for all $\mu \in [z]$. Therefore, if  $a_{j,b} \notin \{a_{i,b} \mid i \in \cF\}$, then node $j \in D$ needs to transmit $S\boldc_{j,b}^\intercal$ during the entire process. Thus, it sends $\dim(S)$ symbols. We have the following lemma.
		\begin{lemma} \label{lem:count} The dimension of $S$ is $\frac{z\zeta^t}{d-k+z}$. 
		\end{lemma}
		See Appendix~\ref{appendix:count} for a proof of Lemma~\ref{lem:count}. 
  Since $z \le h$, it follows that $\dim(S) \le \frac{h\zeta^t}{d-k+h}$, and therefore every node~$j\in D$ transmits at most this much information if~$a_{j,b} \notin \{a_{i,b} \mid i \in \cF\}$, as claimed.
	\end{proof}

	It follows from the minimum distance requirement of $\cU$ that $u^{(j)}_b \in \{u^{(i)}_b \mid i \in \cF\}$ for at most $h(1-\delta)\lambda$ values of $b \in [\lambda]$, for any $j \in D$. Thus, by Proposition~\ref{prop}, each helper node $j \in D$ transmits at most $\frac{h\zeta^{t}}{d-k+h}$ symbols for the repair of at least $\lambda-h(1-\delta)\lambda$ encodings and $\zeta^t$ symbols each for the remaining ones. Therefore, the total download from each helper node is at most 
	\begin{align*}
		& (\lambda-h(1-\delta)\lambda) \frac{h\zeta^t}{d-k+h}+h(1-\delta)\lambda \zeta^t \\ &= \frac{h\lambda \zeta^t}{d-k+h}+ h(1-\delta)\lambda \Big(\zeta^t-\frac{h\zeta^t}{d-k+h}\Big) \\ & =\frac{h\ell}{d-k+h}\Big(1+(1-\delta)(d-k)\Big).   
	\end{align*}	
	Hence, for any valid $(h,d)$,  the code given by Construction~\ref{construction:MDS_simul} is an $(n,k,\ell=\lambda \zeta^t)_{\mathbb{F}}$ MDS array code with $(h,d)$ $\varepsilon$-optimal help-by-transfer property for $\varepsilon=(1-\delta)(d-k)$. Therefore, for  $\varepsilon=(1-\delta)(r-1)$, it is an $(n,k=n-r,\ell=\lambda \zeta^t)_{\mathbb{F}}$ MDS array code with $(h,d)$ $\varepsilon$-optimal help-by-transfer property simultaneously for all $(h,d)$ such that $1 \le h \le r$ and $k \le d \le n-h$.   
	
	\subsection{Existence of  Codes with Logarithmic Sub-packetization Level}
	Similar to Section~\ref{sec:existence}, here we show that sub-packetization level $\ell$ scaling logarithmically with $n$ is possible for codes given by Construction~\ref{construction:MDS_simul}. For $\varepsilon>0$, let $\delta=1-\frac{\varepsilon}{r-1}$, and choose a positive integer $t>\frac{r-1}{\varepsilon}$. 
	 By the GV bound, there exists a code over $[t]$ with block length $\lambda$, relative minimum distance at least $\delta$ and code size  $t^{\lambda(1-h_t(\delta)-o(1))}$. 
	If we pick $\cU$ as this code, then our construction yields an MDS array code with $n= t^{((1-h_t(\delta)-o(1))/\zeta^t)\ell}$. For constant $r$ and $t$, this gives $n=\Omega(\exp(\theta \ell))$, for some constant $\theta$. The field size requirement is $\zeta n+1$. This results in the following theorem. 
	\begin{theorem}\label{thm:simul} 
		Given an integer $r \ge 1$ and a real number $\varepsilon >0$, there is a constant $\theta=\theta(r,\varepsilon)$ such that for infinite values of $\ell$ there exists an $(n=\Omega(\exp(\theta\ell)),k=n-r,\ell)_\mathbb{F}$ MDS array code with $(h,d)$ $\varepsilon$-optimal help-by-transfer property for all $1 \le h \le r$ and $k \le d \le n-h$ simultaneously. Moreover, the field size $|\mathbb{F}|$ scales linearly with $n$. 
	\end{theorem}
Similar to Remark~\ref{remark:AG}, it is possible to get MDS array codes having smaller sub-packetization levels in some cases by choosing $\cU$ as explicit AG codes. 
	\section{Conclusion}
	The $\varepsilon$-MSR codes framework emerged out of the need for MDS codes with small sub-packetization, near-optimal repair bandwidth, as well as the load balancing property. In this paper, we provide the first construction of $\epsilon$-MSR codes such that a failed node can be repaired by accessing an arbitrary set of helper nodes. We also extended this construction to obtain small sub-packetization level MDS array codes that can repair multiple failed nodes with a near-optimal download from each helper node. An interesting future research direction is to explore the possibility of reducing the sub-packetization level of $\varepsilon$-MSR codes even further.

	\printbibliography
	\appendix
	
	\subsection{A Lemma on roots of unity} 
  \label{appendix:alpha} 
 	\begin{lemma} \label{lem:alpha}
		For any two distinct $i,j \in [n]$,  $\alpha_i \alpha_j^{-1}$ is not an $s$-th root of unity, i.e.,  $(\alpha_i \alpha_j^{-1})^s \ne 1$.
	\end{lemma}
 \begin{proof}     
		Note that $(\alpha_i \alpha_j^{-1})^s =(\alpha^{i-1} \alpha^{1-j})^s=\alpha^{(i-j)s}.$ Assume $i>j$. 
		For contradiction, suppose $\alpha^{(i-j)s}=1$. 
		This requires $(q-1) \mid (i-j)s$, since $\alpha$ is a primitive element. This is equivalent to $\frac{q-1}{g}$ dividing $(i-j)$.  But it is not possible, as $\frac{q-1}{g} \ge n$ and $(i-j)<n$. Therefore, for $i>j$, the result holds. 
		If $\alpha^{(i-j)s}=1$, then $(\alpha^{(i-j)s})^{-1}=\alpha^{(j-i)s}=1$. Hence, the lemma statement is true for $j>i$ as well.   \end{proof}
	\subsection{Proof of Lemma~\ref{lem:seq_step}}  \label{appendix:seq_step}
For every~$j\in[\mu-1]$ fix some~$i_j\in \mathcal{F}_j$, which is possible since~$\mathcal{F}_j$ is non-empty by their definition.
	Let $D_{\mu}=D \cup \{i_1,\dots,i_{\mu-1}\}$ and $\overline{D}_{\mu}=[n] \setminus (\{i\}\cup D_{\mu})$.  Let $h(X)\in F[G][X]$ be the annihilator polynomial of the set $\overline{D}_{\mu}$, i.e., 
	$$h(X)=\prod_{j\in \overline{D}_{\mu}}(X-\alpha_jx_{a_{j,b}}).$$ 
	If $d=n-1$, then $\overline{D}_{\mu}$ is empty and we define $h(X)=I$. 
 
Since $a_{i,b}=e_{w_{\mu}}$, 
we have $S_{\mu}\oplus S_{\mu}x_{a_{i,b}}\oplus\ldots\oplus S_{\mu}x_{a_{i,b}}^{s_\mu-1}=\mathbb{F}[G].$ 
 It can be verified that $h(\alpha_{i}x_{a_{i,b}})$ is an invertible matrix. Therefore, together with $\alpha_i \ne 0$, we have that also
	$$S_{\mu}h(\alpha_{i}x_{a_{i,b}})\oplus S_{\mu}(\alpha_ix_{a_{i,b}})h(\alpha_ix_{a_{i,b}})\oplus\ldots\oplus S_{\mu}(\alpha_ix_{a_{i,b}})^{s_\mu-1}h(\alpha_{i}x_{a_{i,b}})=\mathbb{F}[G].$$ 
	Using arguments similar to those used in the proof of the repair property in Theorem~\ref{sec:MDS} it can be shown that, 
	$$S_{\mu}(\alpha_{i}x_{a_{i,b}})^mh(\alpha_{i}x_{a_{i,b}})\boldc_{i,b}^\intercal=
	-\sum_{j\in D_{\mu}}S_{\mu}(\alpha_jx_{a_{j,b}})^mh(\alpha_jx_{a_{j,b}})\boldc_{j,b}^\intercal~~\text{ for } m=0,\ldots,s_\mu-1.$$
	Therefore, if we get enough information from the nodes in $D_{\mu}$ such that we can construct the RHS of the above $s_{\mu}$ equations, then we can recreate $\boldc_{i}$. Fix any $j\in D_{\mu}$. For the $s_\mu$ equations to be constructed, we will need the following information from node $j$:
	$$S_{\mu}(\alpha_jx_{a_{j,b}})^mh(\alpha_jx_{a_{j,b}})\boldc_{j,b}^\intercal \text{ for } m=0,\ldots,s_\mu-1.$$
	If $a_{j,b}\neq a_{i,b}$ then $S_{\mu}$ is an invariant subspace of the operator $x_{a_{j,b}}$, which means that only $S_{\mu}h(\alpha_jx_{a_{j,b}})\boldc_{j,b}^\intercal$ is needed from helper node $j$.  It can be seen that the subspace $S_{\mu}$ is also an invariant subspace under the multiplication by $h(\alpha_jx_{a_{j,b}})$. Therefore, any helper node $j \in D_{\mu}$ with $a_j\neq a_{i}$ only needs to contribute  $S_{\mu}\boldc_{j,b}^\intercal$. 
	\subsection{Proof of Lemma~\ref{lem:count}} \label{appendix:count}
	Recall that $\dim(S)= |\cup_{\mu=1}^{z} G_{\mu}|$. From the definition of  $G_{\mu}$, we get $|G_{\mu}|= \frac{\zeta^t}{d-k+\mu}$ for all $\mu \in [z]$. In  particular, $|G_{1}|= \frac{\zeta^t}{d-k+1}$.  Fix any $\hat{z}$ such that $1 \le \hat{z} \le z-1$. Assume that  $|\cup_{\mu=1}^{\hat{z}} G_{\mu}|  = \frac{\hat{z}\zeta^t}{d-k+\hat{z}}$. To prove the lemma by induction, we show that  $|\cup_{\mu=1}^{\hat{z}+1} G_{\mu}| = \frac{(\hat{z}+1)\zeta^t}{d-k+\hat{z}+1}$. 
	Indeed,
	\begin{align*}
		|\cup_{\mu=1}^{\hat{z}+1} G_{\mu}|= |\cup_{\mu=1}^{\hat{z}} G_{\mu}|+ |G_{\hat{z}+1}|-|\big(\cup_{\mu=1}^{\hat{z}} G_{\mu}\big)\cap G_{\hat{z}+1} |. 
	\end{align*}
	Let $P=\{p \in \{0,1,\dots, \zeta-1\} \mid p \bmod (d-k+\hat{z}+1)=0\}$. It is easy to see that $|P|= \frac{\zeta}{d-k+\hat{z}+1}$.  
	Hence, we have
	 $$|\big(\cup_{\mu=1}^{\hat{z}} G_{\mu}\big)\cap G_{\hat{z}+1} | =  \frac{|\cup_{\mu=1}^{\hat{z}} G_{\mu}|}{d-k+\hat{z}+1}.$$ Therefore,
	\begin{align*}
		|\cup_{\mu=1}^{\hat{z}+1} G_{\mu}| = \frac{\hat{z}\zeta^t}{d-k+\hat{z}} \Big(1-\frac{1}{d-k+\hat{z}+1}\Big)+ \frac{\zeta^t}{d-k+\hat{z}+1}=\frac{(\hat{z}+1)\zeta^t}{d-k+\hat{z}+1}. 	
	\end{align*}
\end{document}